\documentclass[12pt]{article}
\usepackage{amssymb,amsmath,amsthm,latexsym}
\usepackage[pdftex]{color,graphicx}

\newcommand{\comment}[1]{} 

\begin{document}
\newtheorem{theorem}{Theorem}
\newtheorem{proposition}{Proposition}
\newtheorem{lemma}{Lemma}
\newtheorem{definition}{Definition}
\newtheorem{corollary}{Corollary}
\newtheorem{example}{Example}
\newtheorem{remark}{Remark}
\newcommand{\ob}{\mbox{$\overline{\omega}$}}
\newcommand{\om}{\mbox{$\omega$}}
\newcommand{\Tr}{\mbox{Tr}}
\newcommand{\C}{\mbox{$\cal C$}}
\newcommand{\Cperb}{\mbox{$\C^\bot$}}
\newcommand{\ben}{\begin{equation*}}
\newcommand{\een}{\end{equation*}}

\title{A Projection Decoding of a Binary  \\
Extremal Self-Dual Code of Length $40$}
\author{Jon-Lark Kim \\ Department of Mathematics \\ Sogang University  \\ Seoul, 121-742, South Korea \\
{\tt jlkim@sogang.ac.kr } \\ \\
 Nari Lee\\
Department of Mathematics \\
Sogang University \\
Seoul 121-742, South Korea \\
{\tt narirhee@hanmail.net}}



\maketitle

\begin{abstract}
As far as we know, there is no decoding algorithm of any binary self-dual $[40, 20, 8]$ code except for the syndrome decoding applied to the code directly. This syndrome decoding for a binary self-dual $[40,20,8]$ code is not efficient in the sense that it cannot be done by hand due to a large syndrome table. The purpose of this paper is to give two new efficient decoding algorithms for an extremal binary doubly-even self-dual $[40,20, 8]$ code $C_{40,1}^{DE}$  by hand with the help of a Hermitian self-dual $[10,5,4]$ code $E_{10}$ over $GF(4)$. The main idea of this decoding is to project codewords of $C_{40,1}^{DE}$ onto $E_{10}$ so that it reduces the complexity of the decoding of $C_{40,1}^{DE}$. The first algorithm is called the representation decoding algorithm. It is based on the pattern of codewords of $E_{10}$. Using certain automorphisms of $E_{10}$, we show that only eight types of codewords of $E_{10}$ can produce all the codewords of $E_{10}$. The second algorithm is called the syndrome decoding algorithm based on $E_{10}$. It first solves the syndrome equation in $E_{10}$ and finds a corresponding binary codeword of $C_{40,1}^{DE}$.
\end{abstract}



\section{Introduction}
\label{sec:intro}

Self-dual codes have been one of the most active research topics in algebraic coding theory because they have wide connections with mathematical areas including groups, designs, lattices, modular forms, and invariant theory (see~\cite{HufPle},~\cite{JoyKim},~\cite{RS}).

Pless~\cite{PlessSO} started the classification of binary self-dual codes of lengths up to $20$.
There are many papers devoted for the classification of binary self-dual codes (see~\cite{ConSlo},~\cite{Huf2005}). Anguilar-Melchor, et al.~\cite{sd-38-ext} classified all extremal binary self-dual codes of length 38.
Shortly after, it was shown~\cite{BetHarMun2012} that
there are exactly 16470 extremal doubly-even self-dual binary codes
of length 40. Then, it has been shown~\cite{BouBouHar2012} that
there are exactly 10200655 extremal singly-even self-dual codes of length 40.

However, from a viewpoint of decoding, only a few self-dual codes have
efficient decoding algorithms. In particular, we are interested in decoding self-dual codes of reasonable length by hand, i.e., decoding without using computers. For example,  Pless~\cite{Ple_86} showed how to decode the extended binary Golay $[24,12,8]$ code by hand. The $[6,3,4]$ hexacode over $GF(4)$ was used here to decode the extended Golay $[24,12,8]$ code. Later, Gaborit-Kim-Pless~\cite{Kim} suggested two handy decoding algorithms for  three doubly-even self-dual $[32,16,8]$ codes such as
the binary Reed-Muller $[32, 16, 6]$ code and two other doubly-even self-dual $[32,16,8]$ codes denoted by $C83$ (or $2g_{16}$) and $C84$ (or $8f_4$) in the notation of \cite{PlessSO}. As far as we know, there is no efficient decoding algorithm by hand for binary self-dual $[40, 20, 8]$ codes. (The best known binary linear $[40,20]$ code has $d=9$ and has no efficient decoding yet.) For example, the usual syndrome decoding to decode $C_{40,1}^{DE}$  requires $2^{20}$ coset leaders. Hence the comparison of the syndrome of a received vector and the look-up table of $2^{20}$  coset leaders needs much space and time. It has been remarked as a research problem \cite[Research Problem 9.7.8]{HufPle} to find an efficient decoding algorithm that can be used on all self-dual codes or on a large family of such codes. Thus decoding self-dual codes by hand will be an efficient decoding algorithm.

In this paper, we describe two new efficient decoding algorithms for an extremal self-dual $[40,20, 8]$ code $C_{40,1}^{DE}$  by hand with the help of a Hermitian self-dual $[10,5,4]$ code $E_{10}$ over $GF(4)$.
 The algorithm is called the {\em representation decoding algorithm}. It is based on the pattern of codewords of $E_{10}$. Using the automorphisms of $E_{10}$ we show that certain eight types of codewords of $E_{10}$ produce all the codewords of $E_{10}$. The other algorithm is called the {\em syndrome decoding algorithm}. It first solves the syndrome equation in $E_{10}$ and finds a corresponding binary codeword of $C_{40,1}^{DE}$ .
 The main idea of this decoding is to generalize the decoding of the binary Golay code of length 24 whose codewords are arranged as a $4 \times 6$ matrix and are projected onto the $[6,3,4]$ Hexacode over $GF(4)$ (see~\cite{Ple_86}). We give a detailed explanation of these algorithms and show how these algorithms work by showing explicit examples.

\comment{
 We are concerned with extremal {\em singly-even}
self-dual binary codes of length 40.
These codes have minimum distance 8 and have the
following weight enumerator \cite{ConSlo, HarMun2006}.
\begin{equation}
 W(y)=1 +(125 + 16 \beta)y^8 +
         (1664 - 64 \beta)y^{10} + \cdots,
\label{eq:W}
\end{equation}
where $\beta=0,1,2, \cdots, 8, 10$.

All inequivalent singly-even self-dual $[40,20,8]$ binary codes
which have an automorphism of prime order $p \ge 5$ were
constructed in \cite{BuyYor}. There were 39 such codes.
These codes in \cite{BuyYor} have weight enumerators (\ref{eq:W})
for $\beta =$ 0, 1, 5, and 10.

In this paper we construct 55 inequivalent singly-even
self-dual $[40,20,8]$
binary codes which do not have any automorphism of prime order
$p \ge 5$, proving that these are not equivalent
to codes in \cite{BuyYor}.
The codes we construct have weight enumerators (\ref{eq:W})
for $\beta =$ 0, 2, 4, 6, 8, and 10.
Hence there are singly-even self-dual $[40,20,8]$ codes with
$\beta =$ 2, 4, 6, and 8.
Our construction is based on additive self-dual codes
of length 10 over $GF(4)$.
}  

\medskip

\section{Preliminaries}

A {\em linear $[n,k]$ code} over $\mathbb F_q$ is a $k$-dimensional subspace of $\mathbb F_q^n$. The {\em dual} of $\C$ is $\C^{\perp}=\{x \in \mathbb F_q^n~|~ x \cdot c = 0 {\mbox{ for any }} c \in \C\}$, where the dot product is either a usual inner product or a Hermitian inner product. A linear code $\C$ is called {\em self-dual} if $\C=\C^{\perp}$. If $q=2$, then $C$ is called {\em binary}. 
If $q=4$, let $GF(4)=\{0, 1, \om, \ob\}$, where $\ob=\om^2 = \om +1$. It is more natural to consider the Hermitian inner product $\left <, \right>$ on $GF(4)^n$:
for $x=(x_1, x_2, \dots, x_n)$ and $y=(y_1, y_2, \dots, y_n)$ in $GF(4)^n$, $\left <x,y \right>= \sum_{i=1}^n x_i \overline{y_i}$, where $\overline{a}= a^2$.

Besides linear codes over $GF(2)$ or $GF(4)$, we introduce additive codes.
For basic definitions about additive codes, we refer to \cite{All},~\cite{hohn},~\cite{Gab}.
An {\em additive code \C\ over $GF(4)$ of length $n$} is
an additive subgroup of $GF(4)^n$. Since \C\ is a vector
space over $GF(2)$, it has a $GF(2)$-basis consisting of $k~( 0 \le k \le 2n)$
vectors whose entries are in $GF(4)$. We call \C\ an $(n, 2^k)$ code.
A {\em generator matrix} of \C\ is a $k \times n$ matrix with
entries in $GF(4)$ whose rows are a $GF(2)$-basis of \C.\
The {\em weight} wt({\bf c}) of {\bf c} in \C\ is the number of
nonzero components of {\bf c}. The minimum weight $d$ of \C\
is the smallest weight of any nonzero codeword in \C.\ If \C\
is an $(n,2^k)$ additive code of minimum weight $d$, \C\ is
called an $(n,2^k,d)$ code.
In order to define an inner product on additive codes
we define the {\em trace} map,
i.e., for $x$ in $GF(4)$, $\Tr (x)=x+x^2 \in GF(2)$.
We now define a non-degenerate {\em trace inner product} of two vectors
${\bf x}=(x_1x_2\cdots x_n)$ and ${\bf y}=(y_1y_2\cdots y_n)$ in $GF(4)^n$
to be
\[{\bf x}\star{\bf y}=\sum_{i=1}^n\Tr(x_i\overline{y_i}) \in GF(2),\]
where $\overline{y_i}$ denotes the conjugate of $y_i$.
Note that $\Tr(x_i \overline{y_i})=1$ if and only if
$x_i$ and $y_i$ are nonzero distinct elements in $GF(4)$.

If \C\ is an additive code, its {\em dual}, denoted by \Cperb, is the
additive code $\{ {\bf x} \in GF(4)^n \mid {\bf x} \star {\bf c}=0
\mbox{ for all }{\bf c} \in \C\}$. If \C\ is an $(n,2^k)$ code,
then \Cperb\ is an $(n,2^{2n-k})$ code. As usual, \C\ is called
{\em self-dual} if $\C = \Cperb$. We note that if \C\ is self-dual,
\C\ is an $(n,2^n)$ code.  In binary codes there are two types of
self-dual codes which are {\em doubly-even} and {\em singly-even} such that
weights of all the codewords $\equiv 0 \pmod 4$ and  weights of some codewords $\equiv 2 \pmod 4$, respectively. So are in additive codes over $GF(4)$. A self-dual
additive code \C\ is called {\em Type II} (or {\em even}) if all codewords
have even weight, and {\em Type I} if some codeword has odd weight.
It is well known~\cite{CRSS} that a linear self-dual code over $GF(4)$ under
the Hermitian inner product (i.e., a Hermitian self-dual code over $GF(4)$)
is a Type II self-dual additive code over $GF(4)$.


\comment{
In Section~\ref{sec:const} we construct new 55 singly-even self-dual
$[40,20,8]$ binary codes from additive self-dual codes of length 10 over
$GF(4)$. In section~\ref{sec:decoding} we describe how to decode these
singly-even self-dual $[40,20,8]$ codes by using corresponding additive
codes.
}  


H\"{o}hn \cite{hohn} presents two constructions that begin with a self-dual
additive
$(n,2^n)$ code and produce a self-dual binary code of length $4n$.
We describe these two (Constructions A and B) and recall a third (C) from~\cite{Gab}.

If \C\ is a self-dual additive $(n,2^n)$ code over $GF(4)$, let $\widehat{\C}$ be
the binary linear $[4n,n]$ code obtained from \C\ by mapping each
$GF(4)$ component to a 4-tuple in $GF(2)^4$ as follows :
$0 \rightarrow 0000$, $1 \rightarrow 0011$, $\om \rightarrow 0101$,
and $\ob \rightarrow 0110$. Let $d_4$ be the $[4,1,4]$ binary linear
code $\{0000,1111\}$. Let $(d_4^n)_0$ be the $[4n,n-1,8]$ binary
linear code consisting of all codewords of weights divisible by
8 from the $[4n,n,4]$ code $d_4^n$, where $d_4^n$ denotes a direct sum of n copies of $d_4$.
Let  $e_B$ be the $[4n,1]$ code generated by
\[\left\{\begin{array}{rl}
1000\ 1000\ \cdots 1000\ 1000&\mbox{if }n\equiv 0\!\!\!\pmod 4\\
1000\ 1000\ \cdots 1000\ 0111&\mbox{if }n\equiv 2\!\!\!\pmod 4
\end{array}\right.\] and
$e_C$  the binary $[4n,1,n]$ code generated by
$1000\ 1000\ \cdots\ 1000$.

\medskip
\noindent {\em Construction A:} Define $\rho_A(\C)=\widehat{\C}+d_4^n$.

\medskip
\noindent {\em Construction B:}  Assume that $n$ is even.
 Define $\rho_B(\C)=\widehat{\C}+(d_4^n)_0+e_B$.

\medskip

{\noindent} {\em Construction C:} Assume that $n \equiv 2 \pmod4$.
 Define $\rho_C(\C)=\widehat{\C}+(d_4^n)_0+e_C$.

\medskip

We give a brief remark on the classification of additive self-dual  $(10,2^{10},4)$ codes.
Gaborit et al.~\cite{Gab} showed that there are at least $51$ Type I self-dual additive
$(10,2^{10},4)$ codes and at least $5$ Type II self-dual additive $(10,2^{10},4)$ codes.
Later, Bachoc-Gaborit~\cite{BacGab} showed that there are exactly $19$ Type II additive $(10,2^{10},4)$ codes. Finally, Danielsen-Parker~\cite{DanPar} classified that there are exactly $101$ Type I self-dual $(10, 2^{10}, 4)$ codes.

\begin{theorem}
\label{binmap}
Suppose \C\ is a Type II self-dual additive $(10,2^{10},4)$ code.
Then  $\rho_C(\C)$ is a singly-even self-dual $[40,20,8]$ binary code.
Similarly, $\rho_B(\C)$ is a doubly-even self-dual $[40,20,8]$ binary code.
\end{theorem}

\begin{proof}

Note that $\rho_C(\C)=\widehat{\C}+(d_4^n)_0+e_C$ and that $\rho_B(\C)=\widehat{\C}+(d_4^n)_0+e_B$ since $n=10 \equiv 2 \pmod 4$.
It follows from Theorem 5.1 in \cite{Gab} that $\rho_C(\C)$ ($\rho_B(\C)$ respectively) is a singly-even
self-dual (doubly-even self-dual, respectively) binary $[40,20]$ code by construction. Thus  it remains to show the minimum distance
of $\rho_C(\C)$ ($\rho_B(\C)$ respectively) is equal to 8.

For the first statement, the minimum distance
of $\rho_C(\C)$ is less than or equal to 8 by construction.
Let ${\bf a} \in \widehat{\C}$, ${\bf b} \in (d_4^{10})_0$, and
$e_C={\bf c} = 1000\ 1000\ \cdots\ 1000$. It suffices to show that any linear combination of these vectors has weight at least 8.

To do this, for $i=1, \dots, 10$,  let $d_i = \{ 4i-3,4i-2,4i-1,4i \}$ be a subset of coordinates
$\{1,2,3,4, \cdots, 37, 38, 39, 40 \}$. Clearly wt({\bf a}),
wt({\bf b}), and wt({\bf c}) are greater than or equal to 8.
Since wt$({\bf a} +{\bf b}) \ge$ wt({\bf a}) $\ge 8$,
wt$({\bf a} +{\bf b}) \ge 8$.
The weight of ${\bf v} \in \rho_C(\C)$ restricted to
$d_i$ is denoted by wt$({\bf v}{\mid}_{d_i})$.
Since wt$(({\bf a}+{\bf c}) {\mid}_{d_i})=1 {\mbox{ or }}3
{\mbox{ for }} i=1,\cdots, 10$,
we have wt$({\bf a} +{\bf c}) \ge 8$.
Since wt$(({\bf b}+{\bf c}) {\mid}_{d_i})=1 {\mbox{ or }}3
{\mbox{ for }} i=1,\cdots, 10$,
we have wt$({\bf b} +{\bf c}) \ge 8$.
Finally, observe that wt$(({\bf a}+{\bf b}) {\mid}_{d_i}) = 0,
2,{\mbox{ or }}4$, hence that
wt$(({\bf a}+{\bf b} + {\bf c}) {\mid}_{d_i}) =
{\mbox{ wt}}(({\bf a}+{\bf b})+{\bf c}) {\mid}_{d_i})= 1 {\mbox{ or }}3$.
This implies that
wt$({\bf a}+{\bf b} + {\bf c}) \ge 8$.
Therefore the minimum distance of $\rho_C(\C)$ is 8, as desired.

For the second statement, let $e_B={\bf c'} = 1000\ 1000\ \cdots\ 0111$,  ${\bf a} \in \widehat{\C}$, and ${\bf b} \in (d_4^n)_0$.
By construction the minimum distance
of $\rho_B(\C)$ is less than or equal to 8. It suffices to show that any linear combination of these vectors has weight at least 8.

 Let $d_i$  be defined as above.
 Clearly wt({\bf a}),
wt({\bf b}), and $wt({\bf c'})$ are greater than or equal to 8.
Since wt$({\bf a} +{\bf b}) \ge$ wt({\bf a}) $\ge 8$,
wt$({\bf a} +{\bf b}) \ge 8$.
The weight of ${\bf v} \in \rho_B(\C)$ restricted to
$d_i$ is denoted by wt$({\bf v}{\mid}_{d_i})$.
Since wt$(({\bf a}+{\bf c'}) {\mid}_{d_i})=1 {\mbox{ or }}3
{\mbox{ for }} i=1,\cdots, 10$,
we have wt$({\bf a} +{\bf c'}) \ge 8$.
Since wt$(({\bf b}+{\bf c'}) {\mid}_{d_i})=1 {\mbox{ or }}3
{\mbox{ for }} i=1,\cdots 10$,
we have wt$({\bf b} +{\bf c'}) \ge 8$.
Finally, wt$(({\bf a}+{\bf b}) {\mid}_{d_i}) = 0,
2,{\mbox{ or }}4$, hence that
wt$(({\bf a}+{\bf b} + {\bf c'}) {\mid}_{d_i}) =
{\mbox{ wt}}(({\bf a}+{\bf b})+{\bf c'}) {\mid}_{d_i})= 1 {\mbox{ or }}3$.
This implies that
wt$({\bf a}+{\bf b} + {\bf c'}) \ge 8$.
Therefore the minimum distance of $\rho_B(\C)$ is 8, as desired.
\end{proof}


Now, one natural question is whether it is possible to decode some (or all) of these singly-even or doubly-even self-dual $[40,20,8]$ binary codes.

\section{Construction of a doubly-even (or singly-even) $[40,20,8]$ code}
\label{sec:construction}

In this section, we describe
how to construct a singly-even or doubly-even self-dual $[40,20,8]$ binary code which can be decoded easily.

As mentioned above, there are exactly $19$ Type II self-dual additive $(10,2^{10},4)$ codes. Hence by
applying Theorem~\ref{binmap}, we can get  $19$ singly-even self-dual $[40,20,8]$ binary codes and $19$ doubly-even self-dual $[40,20,8]$ binary codes, all of which are inequivalent by Magma. Not all such binary codes can be decoded efficiently.
Since Hermitian self-dual codes over $GF(4)$ are Type II self-dual additive codes, we use the classification of Hermitian self-dual linear $[10,5,4]$ codes. In fact,
there are exactly two non-equivalent Hermitian self-dual linear $[10,5,4]$ codes, denoted by $E_{10}$ and $B_{10}$ in the notation of~\cite{MacOdlSlo}. Their weight enumerator is $W_{10}(y)= 1+ 30y^4 + 300y^6 + 585y^8 + 108 y^{10}$ by \cite{MacOdlSlo}. We rewrite each generator matrix as a $GF(2)$-basis,  where the first five rows form a basis for a Hermitian self-dual linear code over $GF(4)$.

{\tiny
\[ G(E_{10}) = \left [
\begin{array}{cccccccccc}
1 &1 &1 &1 &0 &0 &0 &0 &0 &0 \\
0 &0 &1 &1 &1 &1 &0 &0 &0 &0 \\
0 &0 &0 &0 &1 &1 &1 &1 &0 &0 \\
0 &0 &0 &0 &0 &0 &1 &1 &1 &1 \\
1 &0 &1 &0 &1 &0 &1 &0 &\om &\ob \\
\hline
\om &\om &\om &\om &0 &0 &0 &0 &0 &0 \\
0 &0 &\om &\om &\om &\om &0 &0 &0 &0 \\
0 &0 &0 &0 &\om &\om &\om &\om &0 &0 \\
0 &0 &0 &0 &0 &0 &\om &\om &\om &\om \\
\om &0 &\om &0 &\om &0 &\om &0 &\ob &1\\
\end{array}
\right],~
 G(B_{10}) = \left [ \begin{array}{cccccccccc}
1& 1& 1& 1& 0& 0& 0& 0& 0& 0\\
0& 1& \om& \ob& 1& 0& 0& 0& 0& 0\\
0& 0& 0& 0& 0& 1& 1& 1& 1& 0\\
0& 0& 0& 0& 0& 0& 1& \om& \ob& 1\\
0& 1& \ob& \om& 0& 0& 1& \ob& \om& 0\\
\hline
\om& \om& \om& \om& 0& 0& 0& 0& 0& 0\\
0& \om& \ob& 1& \om& 0& 0& 0& 0& 0\\
0& 0& 0& 0& 0& \om& \om& \om& \om& 0\\
0& 0& 0& 0& 0& 0& \om& \ob& 1& \om\\
0& \om& 1& \ob& 0& 0& \om& 1& \ob& 0\\
\end{array}
\right]
\]
}

Applying Theorem~\ref{binmap}, we obtain the following theorem.

\begin{theorem}
\begin{enumerate}

\item
Each of generator matrices for $\rho_C (E_{10})$  and  $\rho_C(B_{10})$ generates a singly-even self-dual $[40,20,8]$ code.

\item Each of generator matrices for $\rho_B (E_{10})$ and $\rho_B(B_{10})$ generates a doubly-even self-dual $[40,20,8]$ code.
\end{enumerate}

\end{theorem}

Let ${\C}_{40,1}^{DE}=\rho_B (E_{10})$, ${\C}_{40,2}^{DE}=\rho_B (B_{10})$, and
 ${\C}_{40,1}^{SE}=\rho_C (E_{10})$, ${\C}_{40,2}^{SE}=\rho_C (B_{10})$.
Then we describe some properties of ${\C}_{40,i}^{DE}$ and ${\C}_{40,i}^{SE}$ for $i=1,2$. The following computation has been done by Magma.

The weight distribution of ${\C}_{40,1}^{DE}$ is $A_0= 1, A_8=285=A_{32}, A_{12}= 21280=A_{28}, A_{16}=239970=A_{24}, A_{20}=525504$. The order of the automorphism group of ${\C}_{40,1}^{DE}$ is $2^{18} \times 3^2 \times 5 \times 7$. Its covering radius is $8$. The generator matrix $G({\C}_{40,1}^{DE})$ of ${\C}_{40,1}^{DE}$ is given below.

The weight distribution of ${\C}_{40,1}^{SE}$ is $A_0= 1, A_8=285=A_{32}, A_{10}=1024=A_{30}, A_{12}=11040=A_{28}, A_{14}=46080=A_{26}, A_{16}=117090=A_{24}, A_{18}=215040=A_{22}, A_{20}=267456$. We recall that any singly-even self-dual $[40,20,8]$ code has the
following weight enumerator \cite{ConSlo, HarMun2006}.
\begin{equation*}
 W(y)=1 +(125 + 16 \beta)y^8 +
         (1664 - 64 \beta)y^{10} + \cdots,
\label{eq:W}
\end{equation*}
where $\beta=0,1,2, \cdots, 8, 10$. It is useful to know the weight enumerator of ${\C}_{40,2}^{SE}$ in terms of $\beta$. Hence ${\C}_{40,1}^{SE}$ has $W(y)$ with $\beta=10$.
 The order of the automorphism group of ${\C}_{40,1}^{SE}$ is $2^{18} \times 3^2 \times 5$. Its covering radius is $7$. The generator matrix of ${\C}_{40,1}^{SE}$ is obtained by replacing the last row of
 $G({\C}_{40,1}^{DE})$ by the vector $e_C$ from Construction $C$.

Similarly, the weight distribution of ${\C}_{40,2}^{DE}$ is the same as that of $C_{40,1}^{DE}$. The order of the automorphism group of ${\C}_{40,2}^{DE}$ is $2^{16} \times 3^3 \times 5^2$.
Its covering radius is $7$. The generator matrix $G({\C}_{40,2}^{DE})$ of ${\C}_{40,2}^{DE}$ is given below.

The weight distribution of ${\C}_{40,2}^{SE}$ is the same as that of ${\C}_{40,1}^{SE}$.
 The order of the automorphism group of ${\C}_{40,2}^{SE}$ is $2^{16} \times 3^3 \times 5^2$.
 Its covering radius is $7$. The generator matrix of ${\C}_{40,i}^{SE}$ is obtained by replacing the last row of $G({\C}_{40,i}^{DE})$ by the vector $e_C$ in Construction $C$.

{\tiny
\[
G({\C}_{40,1}^{DE})=
\left[
\begin{array}{c}
0 0 1 1 0 0 1 1 0 0 1 1 0 0 1 1 0 0 0 0 0 0 0 0 0 0 0 0 0 0 0 0 0 0 0 0 0 0 0 0\\
0 0 0 0 0 0 0 0 0 0 1 1 0 0 1 1 0 0 1 1 0 0 1 1 0 0 0 0 0 0 0 0 0 0 0 0 0 0 0 0\\
0 0 0 0 0 0 0 0 0 0 0 0 0 0 0 0 0 0 1 1 0 0 1 1 0 0 1 1 0 0 1 1 0 0 0 0 0 0 0 0\\
0 0 0 0 0 0 0 0 0 0 0 0 0 0 0 0 0 0 0 0 0 0 0 0 0 0 1 1 0 0 1 1 0 0 1 1 0 0 1 1\\
0 0 1 1 0 0 0 0 0 0 1 1 0 0 0 0 0 0 1 1 0 0 0 0 0 0 1 1 0 0 0 0 0 1 0 1 0 1 1 0\\
0 1 0 1 0 1 0 1 0 1 0 1 0 1 0 1 0 0 0 0 0 0 0 0 0 0 0 0 0 0 0 0 0 0 0 0 0 0 0 0\\
0 0 0 0 0 0 0 0 0 1 0 1 0 1 0 1 0 1 0 1 0 1 0 1 0 0 0 0 0 0 0 0 0 0 0 0 0 0 0 0\\
0 0 0 0 0 0 0 0 0 0 0 0 0 0 0 0 0 1 0 1 0 1 0 1 0 1 0 1 0 1 0 1 0 0 0 0 0 0 0 0\\
0 0 0 0 0 0 0 0 0 0 0 0 0 0 0 0 0 0 0 0 0 0 0 0 0 1 0 1 0 1 0 1 0 1 0 1 0 1 0 1\\
0 1 0 1 0 0 0 0 0 1 0 1 0 0 0 0 0 1 0 1 0 0 0 0 0 1 0 1 0 0 0 0 0 1 1 0 0 0 1 1\\
1 1 1 1 1 1 1 1 0 0 0 0 0 0 0 0 0 0 0 0 0 0 0 0 0 0 0 0 0 0 0 0 0 0 0 0 0 0 0 0\\
1 1 1 1 0 0 0 0 1 1 1 1 0 0 0 0 0 0 0 0 0 0 0 0 0 0 0 0 0 0 0 0 0 0 0 0 0 0 0 0\\
1 1 1 1 0 0 0 0 0 0 0 0 1 1 1 1 0 0 0 0 0 0 0 0 0 0 0 0 0 0 0 0 0 0 0 0 0 0 0 0\\
1 1 1 1 0 0 0 0 0 0 0 0 0 0 0 0 1 1 1 1 0 0 0 0 0 0 0 0 0 0 0 0 0 0 0 0 0 0 0 0\\
1 1 1 1 0 0 0 0 0 0 0 0 0 0 0 0 0 0 0 0 1 1 1 1 0 0 0 0 0 0 0 0 0 0 0 0 0 0 0 0\\
1 1 1 1 0 0 0 0 0 0 0 0 0 0 0 0 0 0 0 0 0 0 0 0 1 1 1 1 0 0 0 0 0 0 0 0 0 0 0 0\\
1 1 1 1 0 0 0 0 0 0 0 0 0 0 0 0 0 0 0 0 0 0 0 0 0 0 0 0 1 1 1 1 0 0 0 0 0 0 0 0\\
1 1 1 1 0 0 0 0 0 0 0 0 0 0 0 0 0 0 0 0 0 0 0 0 0 0 0 0 0 0 0 0 1 1 1 1 0 0 0 0\\
1 1 1 1 0 0 0 0 0 0 0 0 0 0 0 0 0 0 0 0 0 0 0 0 0 0 0 0 0 0 0 0 0 0 0 0 1 1 1 1\\
1 0 0 0 1 0 0 0 1 0 0 0 1 0 0 0 1 0 0 0 1 0 0 0 1 0 0 0 1 0 0 0 1 0 0 0 0 1 1 1\\
\end{array}
\right],
~~~
G({\C}_{40,2}^{DE})=
\left[
\begin{array}{c}
0 0 1 1 0 0 1 1 0 0 1 1 0 0 1 1 0 0 0 0 0 0 0 0 0 0 0 0 0 0 0 0 0 0 0 0 0 0 0 0\\
0 0 0 0 0 0 1 1 0 1 0 1 0 1 1 0 0 0 1 1 0 0 0 0 0 0 0 0 0 0 0 0 0 0 0 0 0 0 0 0\\
0 0 0 0 0 0 0 0 0 0 0 0 0 0 0 0 0 0 0 0 0 0 1 1 0 0 1 1 0 0 1 1 0 0 1 1 0 0 0 0\\
0 0 0 0 0 0 0 0 0 0 0 0 0 0 0 0 0 0 0 0 0 0 0 0 0 0 1 1 0 1 0 1 0 1 1 0 0 0 1 1\\
0 0 0 0 0 0 1 1 0 1 1 0 0 1 0 1 0 0 0 0 0 0 0 0 0 0 1 1 0 1 1 0 0 1 0 1 0 0 0 0\\
0 1 0 1 0 1 0 1 0 1 0 1 0 1 0 1 0 0 0 0 0 0 0 0 0 0 0 0 0 0 0 0 0 0 0 0 0 0 0 0\\
0 0 0 0 0 1 0 1 0 1 1 0 0 0 1 1 0 1 0 1 0 0 0 0 0 0 0 0 0 0 0 0 0 0 0 0 0 0 0 0\\
0 0 0 0 0 0 0 0 0 0 0 0 0 0 0 0 0 0 0 0 0 1 0 1 0 1 0 1 0 1 0 1 0 1 0 1 0 0 0 0\\
0 0 0 0 0 0 0 0 0 0 0 0 0 0 0 0 0 0 0 0 0 0 0 0 0 1 0 1 0 1 1 0 0 0 1 1 0 1 0 1\\
0 0 0 0 0 1 0 1 0 0 1 1 0 1 1 0 0 0 0 0 0 0 0 0 0 1 0 1 0 0 1 1 0 1 1 0 0 0 0 0\\
1 1 1 1 1 1 1 1 0 0 0 0 0 0 0 0 0 0 0 0 0 0 0 0 0 0 0 0 0 0 0 0 0 0 0 0 0 0 0 0\\
1 1 1 1 0 0 0 0 1 1 1 1 0 0 0 0 0 0 0 0 0 0 0 0 0 0 0 0 0 0 0 0 0 0 0 0 0 0 0 0\\
1 1 1 1 0 0 0 0 0 0 0 0 1 1 1 1 0 0 0 0 0 0 0 0 0 0 0 0 0 0 0 0 0 0 0 0 0 0 0 0\\
1 1 1 1 0 0 0 0 0 0 0 0 0 0 0 0 1 1 1 1 0 0 0 0 0 0 0 0 0 0 0 0 0 0 0 0 0 0 0 0\\
1 1 1 1 0 0 0 0 0 0 0 0 0 0 0 0 0 0 0 0 1 1 1 1 0 0 0 0 0 0 0 0 0 0 0 0 0 0 0 0\\
1 1 1 1 0 0 0 0 0 0 0 0 0 0 0 0 0 0 0 0 0 0 0 0 1 1 1 1 0 0 0 0 0 0 0 0 0 0 0 0\\
1 1 1 1 0 0 0 0 0 0 0 0 0 0 0 0 0 0 0 0 0 0 0 0 0 0 0 0 1 1 1 1 0 0 0 0 0 0 0 0\\
1 1 1 1 0 0 0 0 0 0 0 0 0 0 0 0 0 0 0 0 0 0 0 0 0 0 0 0 0 0 0 0 1 1 1 1 0 0 0 0\\
1 1 1 1 0 0 0 0 0 0 0 0 0 0 0 0 0 0 0 0 0 0 0 0 0 0 0 0 0 0 0 0 0 0 0 0 1 1 1 1\\
1 0 0 0 1 0 0 0 1 0 0 0 1 0 0 0 1 0 0 0 1 0 0 0 1 0 0 0 1 0 0 0 1 0 0 0 0 1 1 1\\
\end{array}
\right]
\]
}



In what follows, we describe how the four codes ${\C}_{40,i}^{DE}$ and ${\C}_{40,i}^{SE}$ for $i=1,2$ are obtained in terms of column and row parities.

To do this, we recall basic terms from \cite{Kim} and \cite{KimK}.

Let ${\bf v}$ be a binary
vector of length 40. We identify it with a $4 \times 10$ array
with zeros and ones in it. Let {\footnotesize$\left[\begin{array}{c}
                          v_1\\
                          v_2 \\
                          v_3\\
                          v_4 \\
                   \end{array}\right]$ } be any column of the array of ${\bf v}$  (where $v_i=0$ or $1$ for $i=1, \ldots,4$).

Label the four rows of the array with the elements of $GF(4): 0, 1, \om, \ob$.      If we take the inner product of a column of our array with the row
labels, we get $0 \cdot v_1 + 1\cdot v_2+\om\cdot v_3+\ob\cdot v_4$ which is an element of $GF(4)$. This defines a linear map called $Proj$ from the set of
binary vectors of length $4m$ to the set of quaternary vectors of length $m$.

For instance, let
${\bf v}$=
(1,1,1,0, 1,0,0,0, 1,1,0,0, 0,1,0,1, 1,0,0,1,
 1,1,0,0, 0,0,1,0, 0,1,0,0, 1,1,1,1, 0,1,1,0)
be the binary vector of length 40. Then

\[ {\bf v}= \begin{array}{cccccccccccc}
 & &1  &2  &3  &4  &5  &6  &7  &8  &9  &10\\ \hline
 0 &   &1  &1  &1  &0  &1  &1  &0  &0  &1  &0 \\
 1 &   &1  &0  &1  &1  &0  &1  &0  &1  &1  &1 \\
\om&   &1  &0  &0  &0  &0  &0  &1  &0  &1  &1 \\
\ob&   &0  &0  &0  &1  &1  &0  &0  &0  &1  &0 \\
 \hline
 & &\ob&0  &1  &\om&\ob&1  &\om&1  &0  &\ob
\end{array}
\]

is projected to the quaternary vector ${\bf w} =(\ob,0,1,\om,\ob,1,\om,1,0,\ob)$
of length 10. Hence we have $Proj({\bf v}) = {\bf w}$.

Let the {\em parity of a column} be either even or odd if an even or
an odd number of ones exist in the column. Define the {\em parity
of the top row} in a similar fashion. Thus columns 1,2,7, and 8
of the above array have odd parity, and the rest have even parity.
The top row has even parity.

Let $S$ be a set of binary vectors of length $4m$ and $\C_4$ a quaternary additive code of length $m$. Then $S$ is said to have $projection$ $O$ $onto$ $\C_4$ if the following conditions are satisfied:
\begin{itemize}
\item[(i)] For any vector ${\bf v}\in S$, $Proj({\bf v})\in \C_4$. Conversely, for any vector ${\bf w}\in \C_4$, all vectors ${\bf v}$ such that
$Proj({\bf v})={\bf w}$ are in $S$.

\item[(ii)] The columns of the array of any vector of $S$ are either all even or all odd.

\item[(iii)] The parity of the top row of the array of any vector of $S$ is the same as the column parity of the array.

\end{itemize}

Using the same notation as above, $S$  is said to have $projection$ $E$ $onto$ $\C_4$  if the conditions (i) and (ii), as well as the following third condition (iii$'$), are satisfied:
\begin{itemize}

\item[(iii$'$)] The parity of the top row of the array of any vector of $S$ is always even.

\end{itemize}

Now we have a result similar to \cite[Lemma 2]{Kim}.

In an analogous way, we obtain the following theorem.

\begin{theorem}
\label{correspondence_de}
Let ${\C}_1=E_{10}$ and ${\C}_2 =B_{10}$.
Let $i=1,2$. Then ${\C}_{40,i}^{DE}$  has projection $O$ onto $\C_i$.
\end{theorem}

\begin{proof}
Let $i=1,2$. Let $\mathcal{I}$ be a set of binary vectors of length 40 having projection $O$ onto $\C_i$. Since we denoted ${\C}_{40,1}^{DE}=\rho_B (E_{10})$, ${\C}_{40,2}^{DE}=\rho_B (B_{10})$ before, we want to show that $\rho_B(\C_i)=\mathcal{I}$.

Note that  $Proj({\bf v})$ in condition (i) of projection $O$ is linear, i.e., if $Proj({\bf v_i})={\bf c_i}$ for  $i=1,2$, then  $Proj({\bf v_1}+{\bf v_2})={\bf c_1}+{\bf c_2}$. Thus any set satisfying condition (i) is linear.  It is easy to see that sum of any two vectors of even(odd, respectively) parities for columns and top row also has even(even, respectively)  parities for columns and top row. Sum of a vector of even parity and a vector of odd parity has odd parities for columns and top row. Thus any set satisfying the condition (ii) and (iii) of projection $O$ is also linear.  Since $\mathcal{I}$  is assummed to have projection $O$ onto $\C_i$,  $\mathcal{I}$ satisfies conditions (i), (ii), and (iii) by the definition of  $projection$ $O$ $onto$ $\C_4$. Thus  $\mathcal{I}$ is a binary linear code.
We note that $\C_i$ as a linear code over $GF(4)$ has an information set which consists of some five linearly independent columns over $GF(4)$ of
$G(\C_i)$. (One such information set is the set of columns  1, 2, 4, 5, and 7 of $G(\C_1)$ (resp. 1, 5, 6, 7, and 10 of $G(\C_2)$) where $G(\C_i)$ is introduced previously.)
So the other columns of $G(\C_i)$ are linear combinations of  the columns in the information set.
Using this, we will compute the size of $\mathcal{I}$.

As before,
we identify a binary vector of length 40 with a $4 \times 10$ array.
First suppose that all columns of our $4 \times 10$ arrays in $\mathcal{I}$
are even and first row is even.  Recall that $\mathcal{I}$ has projection $O$ onto $\C_i$. Then each column in the information set can be even eight times out of 16 choices. There are 2 choices to have even parity for each column outside the information set, except the last one which is automatically determined because of top row parity. Hence there are $8^5\times 2^4\times 1=2^{19}$ vectors when all columns and top row are even. For the same reason we obtain $8^5\times 2^4\times 1=2^{19}$ vectors when all columns and top row are odd. Therefore there are $2\times 2^{19}=2^{20}$ vectors, which implies that $\mathcal{I}$ is a $[40,20]$ linear code. By Theorem \ref{binmap} $\rho_B(\C_i)$  gives a doubly-even self-dual $[40,20,8]$ code. Since $\rho_B(\C_i)$ is a binary linear code of dimension 20 satisfying properties (i), (ii), and (iii) of projection $O$ and $\mathcal{I}$ is the largest linear code of dimension 20 satisfying (i), (ii), and (iii) of projection $O$, we have that $\rho_B(\C_i)=\mathcal{I}$ as desired.
\end{proof}

\begin{remark}
Let ${\C}_1=E_{10}$ and ${\C}_2 =B_{10}$.
Let $i=1,2$. Then ${\C}_{40,i}^{SE}$  has projection $E$ onto $C_i$.
\end{remark}
The proof of the remark is analogous to Theorem \ref{correspondence_de}. Considering only even top parity instead even and odd top row parities, we get the same result.

Since our self-dual codes have minimum distance $d=8$, it can correct up to three errors. In what follows, we show that these can be done very quickly.

\section{Decoding a doubly-even (or singly-even) $[40,20,8]$ code}
\label{sec:decoding}

Because of the simple structure of the generator matrix of $E_{10}$, we focus on the decoding of the doubly-even self-dual binary code ${\C}_{40,1}^{DE}$. The singly-even self-dual binary code ${\C}_{40,1}^{SE}$ is decoded similarly. We will represent all the codewords in $E_{10}$ by only 8 vectors together with certain automorphisms of $E_{10}$.

We call positions $(1,2),(3,4),(5,6),(7,8)$, and $(9,10)$ {\em blocks of $E_{10}$}. In \cite{MacOdlSlo} the order $g$ of the monomial group of $E_{10}$ is $3\cdot 2^{\frac{10}{2}-1}(10/2)!=5760$ with generators $(12)(34),~(13)(24),$ $(13579)(2468~10)$.  It implies that this representation $E_{10}$ is invariant under all permutations of these five blocks and under all even numbers of interchanges within the blocks. So  we can partition all $4^5-1=1023$ vectors of $E_{10}$, except the zero vector,  into  the eight vectors in Table \ref{tab:t1} up to those invariant operations. The following shows how we applied those invariant operations to vectors in $E_{10}$.

\begin{itemize}
\item[(i)] The number of the first type codewords is $10 \times 3$ considering $\frac{5!}{2!3!}=10$ block permutations and scalar multiplication of $\{1, \om, \ob\}$. It is meaningless to consider the interchanges within blocks since the components in a block are identitcal.

\item[(ii)] The number of the second type codewords is $80 \times 3$. Consideirng only block permutations there are $\frac{5!}{4!}=5$ cases. Considering interchanges within 2 blocks there are $\frac{5!}{2!2!}=30$ cases for the type ($01~01~10~10~\om\ob$)  and  $\frac{5!}{3!}=20$ cases for the type ($01~10~10~10~\om\ob$).  Considering interchanges within 4 blocks there are  $\frac{5!}{4!}=5$ cases for the type ($01~01~01~01~\om\ob$) and  $\frac{5!}{3!}=20$ cases for the type ($01~01~01~10~\ob\om$). Considering the scalar multiplication and  the total number of these cases, we get our result.

\item[(iii)] The number of the third type codewords is $\frac{5!}{2!}=60$. For this type we only need to consider block permutations since a codeword with scalar multiplication matches to one of the codewords with block permutations. It is meaningless to consider the interchanges within blocks since the components in a block are identitcal.

\item[(iv)] The number of the fourth type codewords is $5\times 3$ considering $\frac{5!}{4!}=5$ block permutations and scalar multiplications.  It is meaningless to consider the interchanges within blocks since the components in a block are identitcal.

\item[(v)] The number of the fifth type codewords is $30\times 3$  considering $\frac{5!}{2!2!}=30$ block permutations and scalar multiplications.  It is meaningless to consider the interchanges within blocks since the components in a block are identitcal.

\item[(vi)] The number of the sixth type codewords is $160\times 3$. Considering only block permutations there are $\frac{5!}{2!2!}=30$ cases. Considering interchanges within 2 blocks there are $\frac{5!}{3!2!}=10$ cases for the type ($\om\ob~\om\ob~10~10~\om\ob$),  $\frac{5!}{2!}=60$ cases for the type ($\om\ob~\ob\om~01~10~\om\ob$), $\frac{5!}{2!2!}=30$ cases for the type ($\ob\om~\ob\om~01~01~\om\ob$), and  $\frac{5!}{3!}=20$ cases for the type ($\ob\om~\ob\om~01~10~\ob\om$).  Considering interchanges within 4 blocks there are  $\frac{5!}{3!2!}=10$ cases for the type ($\om\ob~\om\ob~01~01~\om\ob$). Considering the scalar multiplication and the total number of these cases, we get our result.

\item[(vii)] The number of the seventh type codewords is $16 \times 3$. Considering only block permutations there are $\frac{5!}{2!3!}=10$ cases. Considering interchanges within 2 blocks there are $\frac{5!}{5!}=1$ case for the type ($\om\ob~\om\ob~\om\ob~\om\ob~\om\ob$)  and  $\frac{5!}{4!}=5$ cases for the type ($\ob\om~\ob\om~\ob\om~\ob\om~\om\ob$). Considering the scalar multiplication and  the total number of these cases, we get our result.

\item[(viii)] The number of the eighth type codewords is $20\times 3$ considering $\frac{5!}{3!}=20$ block permutations and scalar multiplications.  It is meaningless to consider the interchanges within blocks since the components in a block are identical.

\end{itemize}

As an example,
if there is a codeword ($0\ob~\om 1~\om 1~\ob 0~\om 1$), it is easy to find its type  since  its weight is 8 and there are 3 distinct nonzero components. It tells us that this is the sixth type of Table \ref{tab:t1}.
Using this table we can easily check and correct at most three errors and there are many other examples for this in Sec. 4.3.


\begin{table}[!h]
\centering
\caption{All the types of codewords for $E_{10}$}
\label{tab:t1}
\begin{tabular}{cllllllllllll}
\hline
No. & \multicolumn{10}{l}{Type of codewords for $E_{10}$ } & Number of this type & Weight\\
\hline
(i)& 1&1&1&1&0&0&0&0&0&0& $10 \times 3$ &4\\
(ii)& 1&0&1&0&1&0&1&0&$\omega$&$\overline{\omega}$&  $80 \times 3$ &6\\
(iii)& $\omega$&$\omega$&$\overline{\omega}$&$\overline{\omega}$&1&1&0&0&0&0& 60& 6\\
(iv)& 1&1&1&1&1&1&1&1&0&0&  $5\times 3$ &8\\
(v)& 1&1&1&1&$\omega$&$\omega$&$\omega$&$\omega$&0&0&  $30\times 3$&8\\
(vi)& $\overline{\omega}$&$\omega$&$\overline{\omega}$&$\omega$&1&0&1&0&$\omega$&$\overline{\omega}$& $160 \times 3$ & 8\\
(vii)& $\overline{\omega}$&$\omega$&$\overline{\omega}$&$\omega$&$\omega$&$\overline{\omega}$&$\omega$&$\overline{\omega}$&$\omega$&$\overline{\omega}$&  $16\times 3$& 10\\
(viii)& 1&1&1&1&1&1&$\overline{\omega}$&$\overline{\omega}$&$\omega$&$\omega$&  $20 \times 3$&10\\
\hline
\end{tabular}
\end{table}

In what follows, we introduce a lemma in order to describe how many errors and erasures in $E_{10}$ over $GF(4)$ can be corrected. By an erasure, we mean that we know an error position but we do not know a correct value.

\begin{lemma}\label{lemnu}{\rm(\cite[p. 45]{HufPle})}
Let $\C$ be an $[n,k,d]$ code over $GF(q)$. If a codeword ${\bf c}$ is sent and ${\bf y}$ is received where $\upsilon$ errors and $\epsilon$ erasures have occurred, then ${\bf c}$ is the unique codeword in $\C$ closest to ${\bf y}$ provided $2\upsilon  + \epsilon< d$.
\end{lemma}

\begin{table}[!h]
\centering
\caption{All the possible cases of errors that can occur up to three in a vector of  $C_{40,1}^{DE}$ }
\label{tab:tt2}
\begin{tabular}{cll}
\hline
\begin{tabular}{@{}c@{}}The number of errors \\that can occur in $\mathbf{v}\in C_{40,1}^{DE}$\end{tabular}
& \begin{tabular}{@{}c@{}}Parity of \\columns \end{tabular}
& \begin{tabular}{@{}c@{}}The number of errors  \\ in each parity $(x;\bar{y})$\end{tabular}\\
\hline
0 &$[10; \bar{0}]$ & (0;$\bar{0}$)\\
\hline
1 &$[9; \bar{1}]$  & $(0;\bar{1}$) \\
\hline
2  &$[10; \bar{0}]$ &  (2;$\bar{0}$) \\
& $[8; \bar{2}]$  & $(0;\bar{1},\bar{1})$ \\
\hline
3 & $[9; \bar{1}]$  & $(0;\bar{3}$)  \\
 &  $[9; \bar{1}]$  & $(2;\bar{1}$) \\
 &  $[7; \bar{3}]$  & (0;$\bar{1},\bar{1},\bar{1}$) \\
\hline
\end{tabular}
\end{table}

In Sec. 3 we showed how we arranged the binary vectors of length 40
 in columns to do the projection onto $GF(4)^{10}$. Table 2 shows all the possible cases of errors up to ``three'' that we can have  for
binary vectors of length 40.  There is no other case that we can decode to a unique codeword in $C_{40,1}^{DE}$ than the cases listed in Table \ref{tab:tt2}. The notation $[x;\bar{y}]$ we used in the second column of Table 2  denotes that $x$ columns have the same parity and $y$ columns the other and  may assume that  $x \ge y$.  In many cases we can detect errors by checking the column parities since our $[40,20,8]$ code can correct up to three errors. In some cases, however,  column parities do not give enough information about errors since an even number of errors  in a column do not change the  parity. We considered this kind of errors and their combinations as well.

For a binary  vector with  $[x;\bar{y}]$ column parity, the notation ($x_1,\ldots, x_s$; $\bar{y_1},\ldots,\bar{y_t}$) $(0\le s\le x, 0\le t\le \bar{y})$ denotes $x_1,\ldots,x_s$  errors in $s$ columns, and $y_1,\ldots,y_t$  errors  in $t$ columns.
For   (0;$\bar{0}$), it means that there is no error in any of the columns.  For  (0;$\bar{1}$), it means that there is no error in any of nine columns of one parity but an error in the column of the other  parity. This error is easy to find  since one error changes the column parity.  For  (2;$\bar{0}$), it means that  two errors lie in one of the ten  columns. It is easy to see that  these errors  do not change the column parity. For  (0;$\bar{1},\bar{1}$), it means that there is no error in any of eight columns of one parity but an error   each  in  two columns of the other parity.   For  (2;$\bar{1}$), it means that there are two errors in one of the nine columns of one parity  and an error in the  column with the other parity.  For  (0;$\bar{1},\bar{1},\bar{1}$), it means that there are no error in any of seven columns of one parity but an error each in three columns of  the other  parity.

\begin{table}[!h]
\centering
\caption{All the possible cases of errors that can occur for the parity of columns in Table 2 }
\label{tab:t2}
\begin{tabular}{llllll}
\hline
{\mbox{Case}} & \begin{tabular}{@{}c@{}}Prity of \\ columns\end{tabular} & \begin{tabular}{@{}c@{}}The number of errors \\ in $\mathbf{v}\in C_{40,1}^{DE}$\end{tabular} &\begin{tabular}{@{}c@{}}Correct \\ columns\end{tabular}&
 $(\upsilon,\epsilon)$ in $\mathbf{c}\in E_{10}$ & $2 \upsilon + \epsilon < 4 ?$ \\
\hline
{\mbox{I}}    & $[10; \bar{0}]$ & (i) (0;$\bar{0}$)  errors & 10&(0,0) & Yes \\
 & & (ii) (2;$\bar{0}$)& (9) out of 10 &(1,0) & Yes \\
 & & (iii) (4;$\bar{0}$)&(9) out of 10 & (1,0)& Yes\\
 & & (iv) (2,2;$\bar{0}$)& (8) out of 10 &(2,0) & No\\
 & & (v) (2,2,2;$\bar{0}$) & (7) out of 10 &(3,0)& No\\
\hline
{\mbox{II}}   & $[9; \bar{1}]$  & (i) $ (0;\bar{1}$) & 9 &(0,1) & Yes\\
 & & (ii) $(0;\bar{3}$) & 9&(0,1) & Yes \\
 & & (iii) $(2;\bar{1}$)& (8) out of 9 &(1,1)& Yes\\
 & & (iv)$(2,2;\bar{1}$)& (7) out of 9 &(2,1) & No\\
 & & (v) $(2,2,2;\bar{1}$) & (6) out of 9 &(3,1)& No\\
\hline
{\mbox{III}}  & $[8; \bar{2}]$  & (i) $(0;\bar{1},\bar{1}$) &  8 &(0,2) & Yes \\
 & & (ii) $(0;\bar{1},\bar{3})$& 8&(0,2) & Yes \\
 & & (iii) $(2; \bar{1},\bar{1})$&(7) out of 8& (1,2)& No\\
 & & (iv) $(2;\bar{1},\bar{3}$)& (7) out of 8& (1,2) & No \\
& & (v) $(2, 2;\bar{1},\bar{1}$) & (6) out of 8& (2,2)&No\\
\hline
{\mbox{IV}}   & $[7; \bar{3}]$  & (i) (0;$\bar{1},\bar{1},\bar{1}$) &  7 &(0,3) & Yes\\
 & & (ii) $(0;\bar{1},\bar{1},\bar{3}$)&7& (0,3) & Yes \\
 & & (iii) $(2; \bar{1},\bar{1},\bar{1}$)& (6) out of 7& (1,3)& No\\
 & & (iv) $(2;\bar{1},\bar{3},\bar{3}$)& (6) out of 7&(1,3) & No \\
& & (v) $(2;\bar{1},\bar{1},\bar{3}$)& (6) out of 7&(1,3) & No \\
& & (vi) $(2,2;\bar{1},\bar{1},\bar{1}$) & (5) out of 7&(2,3) & No \\
\hline
\end{tabular}
\end{table}


Table 3 shows all the possible cases of errors that can occur for the  parity of columns in Table \ref{tab:tt2}.  In the second column of Table \ref{tab:t2}, we listed all the column parities in Table \ref{tab:tt2}.   The third column of Table \ref{tab:t2} shows all the possible errors that can occur in each case.  However some situations cannot occur such as 4 errors in one column and 1, 2, or 3 in another (i.e., $(4;\bar{1})$, $(4,2;\bar{0})$, $(4;\bar{3})$, repectively). We can directly check from $G(C_{40,1}^{DE})$ that we can always have weight 8 codewords whose $i$-th and $(i+1)$-st columns consist of all ones for $1\le i < 10$. Thus adding these codewords to those errors would correspond to $(0;\bar{3}), (2;\bar{0}),(0;\bar{1})$ case, respectively, giving a coset leader of weight 3, 2, 1, repectively.
The fourth column shows correct columns. For example, the case II(iii)  (8) out of 9 implies that 8 columns out of nine columns with one parity are correct, even though which 8 columns are correct are not decided.

Now we move to the fifth column of Table \ref{tab:t2}. It represents the possible cases of  $\upsilon$ errors and $\epsilon$ erasures that can occur in codewords for $E_{10}$  using Lemma \ref{lemnu}. When a column parity of a codeword for $C_{40,1}^{DE}$ is changed by (an) error(s), then the corresponding component in the codeword for $E_{10}$ is regarded as an erasure. For the case such that  a column parity is not changed by  errors  the corresponding component in the codeword for $E_{10}$ is regarded as an error.

We can get the fifth column  of Table \ref{tab:t2}  from the third column of Table \ref{tab:t2} just by checking the number of numbers without bar and with bar. If the number of errors in $\mathbf{v}\in C_{40,1}^{DE}$ is (0;$\bar{1},\bar{1},\bar{1}$), then the corresponding $(\upsilon,\epsilon)$ in $\mathbf{c}\in E_{10}$ is $(0,3)$ since zero is without bar and three ones are with bar.  If the number of errors in $\mathbf{v}\in C_{40,1}^{DE}$ is ($2;\bar{1}$), then the corresponding $(\upsilon,\epsilon)$ in $\mathbf{c}\in E_{10}$ is $(1,1)$. The sixth column of Table 3 shows that if at most three errors occurred in $C_{40,1}^{DE}$ then
we can obtain a unique vector in $E_{10}$ since all the cases of such errors  satisfy the inequality in Lemma \ref{lemnu}.

Table \ref{tab:tt2} shows all the cases that are decodable to a unique codeword in $C_{40,1}^{DE}$. Table \ref{tab:t2} shows all the kind of errors that can occur for the cases of column parities in Table \ref{tab:tt2}. Any other cases not on the list of Table \ref{tab:t2} definitely have more than three errors, which is regardless in our case since our $C_{40,1}^{DE}$  code can correct up to three errors. We will introduce two kinds of algorithms as follows, one is representation decoding algorithm and the other syndrome docoding algorithm. When consider for  the cases not  in Table \ref{tab:t2},  we declare  it has more than three errors and not decodable on Step 1 in  both algorithms. Thus both algorithms halt for these cases. For the cases in Table \ref{tab:t2}, when we exit from the loop either we have decoded $\bf{v}$  as a codeword of $C_{40,1}^{DE}$  or  declared that $\bf{v}$ has more than three errors.  Executing the if statement at the end of two  algorithms returns the proper value and this satisfies the output condition. Thus both algorithms halt for the cases in Table \ref{tab:t2}.

\subsection{Representation decoding algorithm}

In this section we describe the representation decoding algorithm.
A main idea of this algorithm is that given a received vector ${\bf v}\in GF(2)^{40}$ written as a $4 \times 10$ binary matrix, we project it to a quaternary vector of length 10 and decode this quaternary vector using $E_{10}$ with the help of Table 1 and then obtain a codeword of $C_{40,1}^{DE}$ with the help of Table \ref{tab:t2}.

\medskip

\noindent
{\bf Input:}  A received
vector ${\bf v}\in GF(2)^{40}$ written as a $4 \times 10$ binary matrix.\\
\noindent
{\bf Output:} Either produce a correct codeword of $C_{40,1}^{DE}$ or say that more than three errors occurred.

\noindent
{\bf Representation decoding algorithm:}

\begin{itemize}

\item[Step 1] Compute the parities of the columns of ${\bf v}$ and determine which
case of Table \ref{tab:t2} we are in.  If no case is found, we say that ``more than three errors occurred'' and exit. Otherwise, go to Step 2-I-(i).

\item[Step 2-I-(i)] If the parity of ${\bf v}$ is the case I, then consider it as a case I-(i).   Else, go to Step 2-II-(i).

Compute the projection of ${\bf v}$ onto $GF(4)^{10}$, call it ${\bf y}$. If there is any type in Table 1 matching ${\bf y}$ after the action of the automorphism groups and scalar multiplication, then  we compute the parity of the top row and columns of ${\bf v}$.    If such ${\bf y'}$ is not found,  then we consider it as a case I-(ii) and go to Step2-I-(ii).

 If the parities of columns are all same and top row parity is same as the column parity, we say that ``no errors have occurred'' and ${\bf v}$ is a ``codeword of $C_{40,1}^{DE}$'' and exit.  Else (i.e., case I-(iii)), we say that ``more than three errors occurred'' and exit.

\item[Step 2-I-(ii)] We use Table 1 to find the closest vector ${\bf y'}$ to ${\bf y}$ by changing only one element of ${\bf y}$ after the action of the automorphism groups and scalar multiplication.

     If  ${\bf y'}$ is found, then use ${\bf y'}$ to correct   ${\bf v}$ by changing two elements of ${\bf v}$ in the column where the errors occurred  according to the condition that all the columns and top row have the same parity.

Else (i.e., case I-(iv), (v)), we say that ``more than three errors occurred'' and exit.

After this,  we say ${\bf v}$ is decoded as a ``codeword of $C_{40,1}^{DE}$'' and exit.

\item[Step 2-II-(i)] If  the parity of ${\bf v}$ is the case II, then consider it as a case II-(i).  Else, go to Step 2-III-(i).

Compute the projection of ${\bf v}$ onto $GF(4)^{10}$, call it ${\bf y}$. Since we know the error column position, it is easy to find   the closest vector ${\bf y'}$ to ${\bf y}$  in Table 1  by changing only one element of ${\bf y}$ after the action of the automorphism groups and scalar multiplication.  If such ${\bf y'}$ is not found (i.e., case II-(iii), (iv), (v)), go to Step II-(iii).

We use ${\bf y'}$ to correct ${\bf v}$ by changing  one element in the  column where the error
occurred  according to the condition that all the columns and top row have the same parity.
 Then we say ${\bf v}$ is decoded as a ``codeword of $C_{40,1}^{DE}$'' and exit. Else,   go to Step 2-II-(ii).

 \item[Step 2-II-(ii)]   Now we use ${\bf y'}$ to correct ${\bf v}$ by changing  three elements in the  column where the errors
occurred  according to the condition that all the columns and top row have the same parity. Then  we say  ${\bf v}$ is decoded as a ``codeword of $C_{40,1}^{DE}$'' and exit.

 \item[Step 2-II-(iii)]  We use Table 1 to find the closest vector ${\bf y'}$ to ${\bf y}$  by changing two elements of ${\bf y}$ after the action of the automorphism groups and scalar multiplication.

     If  ${\bf y'}$ is found, then we use ${\bf y'}$ to correct ${\bf v}$ by changing  two elements in two   columns where the errors
occurred  according to the condition that all the columns and top row have the same parity.   Else (case II-(iv), (v)), we say that ``more than three errors occurred'' and exit.

After this,  we say  ${\bf v}$ is decoded as a ``codeword of $C_{40,1}^{DE}$'' and exit.

\item[Step 2-III-(i)] If the parity of ${\bf v}$ is the case III, then consider it as a case III-(i).  Else, go to Step 2-IV-(i).

Compute the projection of ${\bf v}$ onto $GF(4)^{10}$, call it ${\bf y}$. Since we know the error column positions, it is easy to find   the closest vector ${\bf y'}$ to ${\bf y}$  in Table 1  by changing two elements of ${\bf y}$ after the action of the automorphism groups and scalar multiplication. If such ${\bf y'}$ is not found (i.e., case III-(iii), (iv), (v)), we say that ``more than three errors occurred'' and exit.

We use ${\bf y'}$ to correct ${\bf v}$ by changing  two elements in two  columns where the errors
occurred  according to the condition that all the columns and top row have the same parity.
 Then  we say  ${\bf v}$ is decoded as a ``codeword of $C_{40,1}^{DE}$'' and  exit. Else, we are in (case III-(ii)), so we say that ``more than three errors occurred'' and exit.

\item[Step 2-IV-(i)] If the parity of ${\bf v}$ is the case IV, then consider it as a case IV-(i).  Compute the projection of ${\bf v}$ onto $GF(4)^{10}$, call it ${\bf y}$. Since we know the error column positions, it is easy to find   the closest vector ${\bf y'}$ to ${\bf y}$  in Table 1  by changing three elements of ${\bf y}$ after the action of the automorphism groups and scalar multiplication. If such ${\bf y'}$ is not found (i.e., case IV-(iii), (iv), (v), (vi)), we say that ``more than three errors occurred'' and exit.

We use ${\bf y'}$ to correct ${\bf v}$ by changing  three elements in three  columns where the errors
occurred  according to the condition that all the columns and top row have the same parity.
 Then we say  ${\bf v}$ is decoded as a ``codeword of $C_{40,1}^{DE}$'' and exit. Else, we are in (case IV-(ii)), so we say that ``more than three errors occurred'' and exit.

\end{itemize}

The following columns yield the elements in $GF(4)$.
\begin{displaymath}
0 \in GF(4) \Rightarrow \left( \begin{array}{c}
0\\
0\\
0\\
0
\end{array}\right)
\left( \begin{array}{c}
1\\
1\\
1\\
1\\
\end{array}\right)
\left( \begin{array}{c}
1\\
0\\
0\\
0\\
\end{array}\right)
\left( \begin{array}{c}
0\\
1\\
1\\
1\\
\end{array}\right),
~~~~ 1\in GF(4) \Rightarrow \left( \begin{array}{c}
1\\
1\\
0\\
0\\
\end{array}\right)
\left( \begin{array}{c}
0\\
0\\
1\\
1\\
\end{array}\right)
\left( \begin{array}{c}
0\\
1\\
0\\
0\\
\end{array}\right)
\left( \begin{array}{c}
1\\
0\\
1\\
1\\
\end{array}\right),
\end{displaymath}

 \begin{displaymath}
\omega\in GF(4) \Rightarrow \left( \begin{array}{c}
1\\
0\\
1\\
0
\end{array}\right)
\left( \begin{array}{c}
0\\
1\\
0\\
1\\
\end{array}\right)
\left( \begin{array}{c}
0\\
0\\
1\\
0\\
\end{array}\right)
\left( \begin{array}{c}
1\\
1\\
0\\
1\\
\end{array}\right),
~~~~ \bar{\omega}\in GF(4) \Rightarrow \left( \begin{array}{c}
1\\
0\\
0\\
1\\
\end{array}\right)
\left( \begin{array}{c}
0\\
1\\
1\\
0\\
\end{array}\right)
\left( \begin{array}{c}
0\\
0\\
0\\
1\\
\end{array}\right)
\left( \begin{array}{c}
1\\
1\\
1\\
0\\
\end{array}\right).
\end{displaymath}
As we can see above, we have four choices for each element in $GF(4)$ and two choices for even parity and two for odd.  Thus after  the correct projection $\bf{y'}$  is found, we change elements in columns of  $\bf{v}$ to satisfy $Proj(\bf{v})=\bf{y'}$. Considering the column and top row parity, we are left with only one choice out of four. After changing elements according to the steps in the representation decoding algorithm,  we can say $\bf{v}$ is decoded to a unique vector in  $C_{40,1}^{DE}$.

\subsection{Syndrome decoding algorithm}

In this section, we give a syndrome decoding algorithm based on $E_{10}$.
A main idea of this algorithm is that given a received vector ${\bf v}\in GF(2)^{40}$ written as a $4 \times 10$ binary matrix, we project it to a quaternary vector of length 10 and decode this quaternary vector using the syndrome decoding of $E_{10}$ and then obtain a codeword of $C_{40,1}^{DE}$ with the help of Table \ref{tab:t2}.

Let $H$ be a parity check matrix of $E_{10}$ consisting of the first five rows of $G(E_{10})$ in Section 3.  We use the first five rows of $H=G(E_{10})$ since we are considering it as a parity check matrix of a linear self-dual code over $GF(4)$.

\medskip

\noindent
{\bf Input:} A vector ${\bf v}\in GF(2)^{40}$ as a $4 \times 10$ binary matrix.

\noindent
{\bf Output:} Either produce a correct codeword of $C_{40,1}^{DE}$ or say that more than three errors occurred.

\noindent
{\bf Syndrome decoding algorithm algorithm:}

\begin{itemize}

\item[Step 1] Compute the parities of the columns of ${\bf v}$ and determine which
case of Table \ref{tab:t2} we are in.  If no case is found, we say that ``more than three errors occurred'' and exit. Otherwise, go to Step  2-I-(i).

\item[Step 2-I-(i)] If the parity of ${\bf v}$ is the case I, then consider it as a case I-(i).  Else, go to Step 2-II-(i).

Compute the projection of ${\bf v}$ onto $GF(4)^{10}$, call it ${\bf y}$. If the syndrome of ${\bf y}$ is zero, then we compute the parity of the top row and columns of ${\bf v}$. Else, go to Step2-I-(ii). When we get the syndrome of ${\bf y}$, we use  $H$ as a  parity check matrix  of $E_{10}$.

 If the parities of columns are all same and top row parity is the same as the column parity,
we say that ``no errors have occurred'' and exit. Else (i.e., case I-(iii)), we say that ``more than three errors occurred'' and exit.

\item[Step 2-I-(ii)] If the syndrome is a  scalar multiple $k$ of an $i$-th column of the parity check matrix $H$, then the $i$-th coordinate of the error vector ${\bf e}$ is $k$ and elsewhere zeros.   Else (i.e., case I- (iv), (v)), we say that ``more than three errors occurred'' and exit.

By adding ${\bf e}$ to ${\bf y}$, we get the correct projection ${\bf y'}$.  Now we use ${\bf y'}$ to correct ${\bf v}$ by changing  two elements in the  column where the errors
occurred  according to the condition that all the columns and top row have the same parity. Then we say ${\bf v}$ is decoded as a ``codeword of $C_{40,1}^{DE}$'', then exit.

\item[Step 2-II-(i)] If the parity of ${\bf v}$ is the case II, then consider it as a case II-(i).   Else, go to Step 2-III-(i).

Compute the projection of ${\bf v}$ onto $GF(4)^{10}$, call it ${\bf y}$.  If the syndrome is a  scalar multiple $k$ of an $i$-th column of $H$, then the $i$-th coordinate of the error vector ${\bf e}$ is $k$ and elsewhere zeros. Else, go to Step 2-II-(iii).

By adding ${\bf e}$ to ${\bf y}$, we get the correct projection ${\bf y'}$.  We use ${\bf y'}$ to correct ${\bf v}$ by changing  one element in the  column where the error
occurred  according to the condition that all the columns and top row have the same parity.
 Now we say  ${\bf v}$ is decoded as a ``codeword of $C_{40,1}^{DE}$'', then exit. Else, go to Step 2-II-(ii).

 \item[Step 2-II-(ii)]   We use ${\bf y'}$ to correct ${\bf v}$ by changing  three elements in the  column where the errors
occurred  according to the condition that all the columns and top row have the same parity.
 Now we say  ${\bf v}$ is decoded as a ``codeword of $C_{40,1}^{DE}$'', then exit.

 \item[Step 2-II-(iii)]  If the syndrome is written  as a linear combination
of two  columns in $H$, then we can get an error vector and obtain a correct projection  ${\bf y'}$. Else (i.e, case II-(iv), (v)), we say that ``more than three errors occurred'' and exit.

Since we know the correct projection ${\bf y'}$ now,  we use ${\bf y'}$ to correct ${\bf v}$ by changing  three elements in three  columns  where the errors
occurred  according to the condition that all the columns and top row have the same parity. Now we say ${\bf v}$ is decoded as a ``codeword of $C_{40,1}^{DE}$'', then exit.

\item[Step 2-III-(i)] If the parity of ${\bf v}$ is the case III,  consider it as a case III-(i).  Else, go to Step 2-IV-(i).

Compute the projection of ${\bf v}$ onto $GF(4)^{10}$, call it ${\bf y}$.  If the syndrome is written  as a linear combination of two  columns in  $H$, then we can get an error vector and obtain a correct projection  ${\bf y'}$. Else (i.e., case III-(iii), (iv), (v)), we say that ``more than three errors occurred'' and exit.

We use ${\bf y'}$ to correct ${\bf v}$ by changing  two elements in two  columns  where the errors
occurred  according to the condition that all the columns and top row have the same parity.
 Now we say  ${\bf v}$ is decoded as a ``codeword of $C_{40,1}^{DE}$'', then exit. Else (i.e, case III-(ii)), we say that ``more than three errors occurred'' and exit.

\item[Step 2-IV-(i)] If the parity of ${\bf v}$ is the case IV, then consider it as a case IV-(i). Compute the projection of ${\bf v}$ onto $GF(4)^{10}$, call it ${\bf y}$.  If the syndrome is written  as a combination of three  columns in  $H$, then we can get an error vector and obtain a correct projection  ${\bf y'}$. Else  (i.e., case IV-(iii), (iv), (v), (vi)),  we say that ``more than three errors occurred'' and exit.

We use ${\bf y'}$ to correct ${\bf v}$ by changing  three elements in three  columns where the errors
occurred  according to the condition that all the columns and top row have the same parity.
 Now we say   ${\bf v}$ is decoded as a ``codeword of $C_{40,1}^{DE}$'', then exit. Else  (i.e., case IV-(ii)), we say that ``more than three errors occurred'' and exit.
\end{itemize}


To see this algorithm more clearly we give several examples.

\subsection{Examples}
\label{subsec:example}

The steps in these examples are the ones in
Section 4.2.

\begin{example}[Case I]
\label{CaseI} All the columns have the same parity.\\

\[
{\setlength{\arraycolsep}{1.785pt}
\renewcommand{\arraystretch}{.5}
{\bf v} = \begin{array}{cccccccccccc}
        &  &  1&  2&  3&  4&  5&  6&  7&  8&  9& 10\\ \hline
       0&  &  0&  1&  1&  0&  1&  1&  1&  1&  1&  0\\
       1&  &  1&  0&  0&  1&  0&  0&  0&  0&  1&  0 \\
     \om&  &  0&  0&  1&  1&  1&  0&  0&  1&  0&  1 \\
     \ob&  &  0&  0&  1&  1&  1&  0&  0&  1&  1&  0 \\
      \hline
 {\bf y}&  &  1&  0&  1&  0&  1&  0&  0&  1&\om&\om \\
      \end{array}
}\]

{\em  
{\bf{(a) Representation Decoding }}
\begin{itemize}
\item[Step 1] Compute the parities of the columns of ${\bf v}$ and determine which
case of Table \ref{tab:t2} we are in. Since all the columns have one parity, this is the case I.  Thus  we move to the Step 2-I-(i).

\item[Step 2-I-(i)]  Compute the projection of ${\bf v}$ onto $GF(4)^{10}$, call it ${\bf y}$. Note that there is no type matching ${\bf y}$ in Table 1 after the action of the automorphism groups and scalar multiplication since ${\bf y}$ has four zeros in four
blocks each and the same non-zero elements of $GF(4)$ in a block. Thus we move to Step 2-I-(ii).

\item[Step 2-I-(ii)] We use Table 1 to find the closest vector ${\bf y'}$ to ${\bf y}$  by changing only one element of ${\bf y}$ after the action of the automorphism groups and scalar multiplication. We can find a unique $E_{10}$ codeword closest to ${\bf y}$ of weight 6 by correcting only one element, i.e.
 $(1,0,1,0,1,0,0,1,\mbox{\boldmath \ob},\om)$ (see Type 2 in Table \ref{tab:t1}).

Now we use ${\bf y'}$ to correct   ${\bf v}$ by changing two elements of ${\bf v}$ in the column
where the errors occurred  according to the condition that all the columns and top row have the same parity.

Finally ${\bf v}$ is decoded as a ``codeword of $C_{40,1}^{DE}$''.  We uniquely decode ${\bf v}$ as follows:

\[
{\setlength{\arraycolsep}{1.785pt}
\renewcommand{\arraystretch}{.5}
{\bf v} = \begin{array}{cccccccccccc}
        &  &  1&  2&  3&  4&  5&  6&  7&  8&  9& 10\\ \hline
       0&  &  0&  1&  1&  0&  1&  1&  1&  1&  1&  0\\
       1&  &  1&  0&  0&  1&  0&  0&  0&  0&  1&  0 \\
     \om&  &  0&  0&  1&  1&  1&  0&  0&  1&  {\bf 1}&  1 \\
     \ob&  &  0&  0&  1&  1&  1&  0&  0&  1&  {\bf 0}&  0 \\
      \hline
 {\bf y'}&  &  1&  0&  1&  0&  1&  0&  0&  1&{\boldmath \ob}&\om \\
      \end{array}.
}\]

Then we exit.

\end{itemize}

{\bf{(b) Syndrome Decoding }}
\begin{itemize}

\item[Step 1] Same as the Step 1 in representation decoding.

\item[Step 2-I-(i)]  Compute the projection of ${\bf v}$ onto $GF(4)^{10}$, call it ${\bf y}$. Since  the syndrome of ${\bf y}$ is not zero as below, we move to the  Step2-I-(ii).

\[
{\setlength{\arraycolsep}{1.785pt}
\renewcommand{\arraystretch}{.5}
H {\bf \overline{y}}^T =  \left[\begin{array}{c}
                          0\\
                          0 \\
                          0 \\
                          1 \\
                          \om
                   \end{array}\right].
}\]

\item[Step 2-I-(ii)] Since the syndrome is the ninth column of $H$ (see $G(E_{10})$  in Sec. 3),  we can get our error vector  ${\bf e} = (0,0,0,0,0,0,0,0,1,0)$ giving a codeword
${\bf y'}={\bf y}+{\bf e}=(1,0,1,0,1,0,0,1,\mbox{\boldmath \ob},\om).$

 Now we can use ${\bf y'}$ to correct   ${\bf v}$ by changing two elements of ${\bf v}$ in the column
where the errors occurred  according to the condition that all the columns and top row have the same parity.

Finally ${\bf v}$ is decoded as a ``codeword of $C_{40,1}^{DE}$''. Then we exit.

\end{itemize}

} 

\end{example}

\begin{example}[Case II]
\label{CaseII} Nine columns have one parity, and one left has the other parity.

\[
{\setlength{\arraycolsep}{1.785pt}
\renewcommand{\arraystretch}{.5}
{\bf v} = \begin{array}{cccccccccccc}
        &  &  1&  2&  3&  4&  5&  6&  7&  8&  9& 10\\ \hline
       0&  &  1&  0&  1&  1&  1&  0&  1&  1&  1&  0\\
       1&  &  0&  1&  1&  0&  0&  0&  0&  0&  1&  0 \\
     \om&  &  1&  1&  1&  1&  1&  1&  1&  0&  0&  0 \\
     \ob&  &  1&  1&  0&  1&  0&  0&  1&  0&  1&  1 \\
      \hline
 {\bf y}&  &  1&  0&  \ob&  1&  \om&  \om&  1&  0&\om&\ob \\
      \end{array}
}\]

{\em
{\bf{(a) Representation Decoding }}
\begin{itemize}

\item[Step 1] Compute the parities of the columns of ${\bf v}$ and determine which
case of Table \ref{tab:t2} we are in. Since nine columns have one parity and one left has the other parity, this is the case II.  Thus  we move to the Step 2-II-(i).

\item[Step 2-II-(i)] If  the parity of ${\bf v}$ is the case II, then consider it as a case II-(i).  Compute the projection of ${\bf v}$ onto $GF(4)^{10}$, call it ${\bf y}$. We cannot find the closest ${\bf y'}$ to ${\bf y}$  in Table 1  by changing only one element of ${\bf y}= (1,0,\ob, 1,-,\om,1,0,\om,\ob)$ after the action of the automorphism groups and scalar multiplication  since it does not match to any of  $E_{10}$ codewords of weight 8  with three distinct nonzero elements .  It implies that one column among nine odd columns of ${\bf y}$ is not correct. Thus we move to the Step 2-II-(iii).

 \item[Step 2-II-(iii)]  We use Table 1 to find the closest vector ${\bf y'}$ to ${\bf y}$  by changing two elements of ${\bf y}$ after the action of the automorphism groups and scalar multiplication and find the unknown error position. In this case, it is the fourth position.  Now we change it to $\om$ (Type 6 of Table \ref{tab:t1}).   Now we use ${\bf y'}$ to correct ${\bf v}$ by changing  three elements in two  columns where the errors
occurred  according to the condition that all the columns and top row have the same parity. Thus ${\bf v}$ is uniquely decoded as a ``codeword of $C_{40,1}^{DE}$''.

\[
{\setlength{\arraycolsep}{1.785pt}
\renewcommand{\arraystretch}{.5}
{\bf v} = \begin{array}{cccccccccccc}
        &  &  1&  2&  3&  4&  5&  6&  7&  8&  9& 10\\ \hline
       0&  &  1&  0&  1&  1&  1&  0&  1&  1&  1&  0\\
       1&  &  0&  1&  1&  {\bf 1}&  {\bf 1}&  0&  0&  0&  1&  0 \\
     \om&  &  1&  1&  1&  {\bf 0}&  1&  1&  1&  0&  0&  0 \\
     \ob&  &  1&  1&  0&  1&  0&  0&  1&  0&  1&  1 \\
      \hline
 {\bf y'}&  &  1&  0&  \ob&  {\boldmath \om}&  {\boldmath \ob}&  \om&  1&  0&\om&\ob \\
      \end{array}
}\]

Then we exit.

\end{itemize}

{\bf{(b) Syndrome Decoding }}
\begin{itemize}

\item[Step 1] Same as the Step 1 in representation decoding.

\item[Step 2-II-(i)] If the parity of ${\bf v}$ is the case II, then consider it as a case II-(i).  Compute the projection of ${\bf v}$ onto $GF(4)^{10}$, call it ${\bf y}$. Since the fifth column of  ${\bf v}$ is the only one with odd parity, we check whether the syndrome $H {\bf \overline{y}}^T$ is the scalar multiple of  fifth column of $H$ or not. In this case,  the syndrome is not a scalar multiple of column of $H$ as below. We let  ${\bf e}=[e_{1},e_{2},\ldots, e_{10}]$ denote the error vector and $\overline{e}_i$  a conjugate of $e_i$.


\[
{\setlength{\arraycolsep}{1.785pt}
\renewcommand{\arraystretch}{.5}
 H {\bf \overline{y}}^T = \left[\begin{array}{c}
                          \om\\
                          \ob \\
                          1 \\
                          0 \\
                          1
                   \end{array}\right]
                 \ne \overline{e}_5 \left[\begin{array}{c}
                          0 \\
                          1 \\
                          1 \\
                          0 \\
                          1
                   \end{array}\right].
}\]

So we move to the Step 2-II-(iii).

 \item[Step 2-II-(iii)]  We consider the syndrome as a linear combination of the fifth column and one of the remaining nine columns. We let ${\bf S_i}=[s_{i_1},s_{i_2},s_{i_3},s_{i_4},s_{i_5}]^T$ denote the $i$-th column of $H$.

\[
{\setlength{\arraycolsep}{1.785pt}
\renewcommand{\arraystretch}{.5}
 H {\bf \overline{y}}^T = \left[\begin{array}{c}
                          \om\\
                          \ob \\
                          1 \\
                          0 \\
                          1
                   \end{array}\right]
                 = \overline{e}_5 \left[\begin{array}{c}
                          0 \\
                          1 \\
                          1 \\
                          0 \\
                          1
                   \end{array}\right]
                 +\overline{e}_i S_i
= H {\bf \overline{e}}^T.
}\]
Since the first entry of the syndrome is $\omega$ and the first entry of the fifth column of $H$ is 0, the first entry of $S_i$ must be non-zero. So $i$ should be one of 1,2,3, or 4. From the fact that  non-zero first entries are all 1, we get $e_i$ as follows:
\begin{equation*}
\begin{split}
\omega &= 0\cdot \bar{e}_5+1\cdot \bar{e}_i,\\
\omega &= \bar{e}_i, \\
\Rightarrow e_i &=\bar{\omega}.\\
\end{split}
\end{equation*}
Now we check the third entries on both sides. Since the third entries of $S_1,~S_2,~S_3,$ and $S_4$ are all zeros, we get
\begin{equation*}
\begin{split}
1 &= 1\cdot \bar{e}_5+ 0\cdot \om,\\
\Rightarrow \bar{e}_5&=1.\\
\end{split}
\end{equation*}
Next we check the second entries.

\begin{equation*}
\begin{split}
\bar{\omega} &= 1\cdot 1+ s_{i_2}\cdot \omega, \\
\Rightarrow s_{i_2}&=1.\\
\end{split}
\end{equation*}

For the last we check the fifth entries.

\begin{equation*}
\begin{split}
1 &= 1\cdot 1+ s_{i_5} \cdot \om,\\
\Rightarrow s_{i_5}&=0.\\
\end{split}
\end{equation*}

Thus  ${\bf S_i}=[s_{i_1},1,s_{i_3},s_{i_4},0]^T$ and this tells us $i=4$, i.e., $S_4 = (1,1,0,0,0)^T$.

Thus
we have the codeword
\ben
\begin{split}
{\bf y'}&={\bf y}+{\bf e}\\
&=(1,0,\ob,1,\om,\om,1,0,\om,\ob)+(0,0,0,\ob,1,0,0,0,0,0)\\
&=(1,0,\ob,\mbox{\boldmath \om},\mbox{\boldmath \ob},\om,1,0,\om,\ob).
\end{split}
\een

Next we use ${\bf y'}$ to correct ${\bf v}$ by changing  three elements in two  columns where the errors
occurred  according to the condition that all the columns and top row have the same parity. Then we exit.

\end{itemize}
}
\end{example}

\begin{example}[Case III]
\label{CaseIII} Eight  columns have one parity and remaining two columns
 have the other parity.

\[
{\setlength{\arraycolsep}{1.785pt}
\renewcommand{\arraystretch}{.5}
{\bf v} = \begin{array}{cccccccccccc}
        &  &  1&  2&  3&  4&  5&  6&  7&  8&  9& 10\\ \hline
       0&  &  1&  1&  1&  0&  0&  1&  1&  1&  1&  0\\
       1&  &  1&  1&  1&  0&  1&  0&  0&  0&  0&  1 \\
     \om&  &  0&  0&  1&  0&  0&  1&  1&  1&  0&  1 \\
     \ob&  &  1&  1&  0&  1&  1&  0&  1&  1&  0&  1 \\
      \hline
 {\bf y}&  &  \om&  \om&\ob&\ob& \om&\om&  1&1&0&  0 \\
      \end{array}
}\]

{\em
{\bf{(a) Representation Decoding }}
\begin{itemize}
\item[Step 1] Compute the parities of the columns of ${\bf v}$ and determine which
case of Table \ref{tab:t2} we are in. Since eight columns have one parity and remaining  two
columns have the other parity, this is the case III.  Thus  we move to the Step 2-III-(i).

\item[Step 2-III-(i)] If the parity of ${\bf v}$ is the case III, then consider it as a case III-(i).  Compute the projection of ${\bf v}$ onto $GF(4)^{10}$, call it ${\bf y}$. Since we know the error column positions, it is easy to find   the closest vector ${\bf y'}$ to ${\bf y}$  in Table 1  by changing two elements of ${\bf y}=(\om,\om,\ob,\ob,-,-,1,1,0,0)$ after the action of the automorphism groups and scalar multiplication.

Since there are three distinct nonzero elements in ${\bf y}$ and each block has identical elements, we cannot find a matching type for weight 8 in Table 1.
Now we try with weight 6 (Type 3 of Table 1).  Finally we find ${\bf y'}=(\om,\om,\ob,\ob,{\bf 0},{\bf 0},1,1,0,0)$ which is the closet vector to ${\bf y}$ after the action of the automorphism groups and scalar multiplication.

We use ${\bf y'}$ to correct ${\bf v}$ by changing  two elements in two  columns where the errors
occurred  according to the condition that all the columns and top row have the same parity
 as follows :

\[
{\setlength{\arraycolsep}{1.785pt}
\renewcommand{\arraystretch}{.5}
{\bf v} = \begin{array}{cccccccccccc}
        &  &  1&  2&  3&  4&  5&  6&  7&  8&  9& 10\\ \hline
       0&  &  1&  1&  1&  0&  0&  1&  1&  1&  1&  0\\
       1&  &  1&  1&  1&  0&  1&  0&  0&  0&  0&  1 \\
     \om&  &  0&  0&  1&  0&  {\bf 1}&  {\bf 0}&  1&  1&  0&  1 \\
     \ob&  &  1&  1&  0&  1&  1&  0&  1&  1&  0&  1 \\
      \hline
 {\bf y'}&  &  \om&  \om&\ob&\ob& {\bf 0}&{\bf 0}&  1&1&0&  0 \\
      \end{array}.
}\]

Now  ${\bf v}$ is decoded as a ``codeword of $C_{40,1}^{DE}$''. We exit.

\end{itemize}

{\bf{(b) Syndrome Decoding }}
\begin{itemize}

\item[Step 1] Same as the Step 1 in representation decoding.

\item[Step 2-III-(i)] If the parity of ${\bf v}$ is the case III $(8,\bar{2})$, then consider it as a case III-(i).  Compute the projection of ${\bf v}$ onto $GF(4)^{10}$, call it ${\bf y}$.  Since we know the fifth and sixth columns of ${\bf v}$ are the only ones with even parity, we know these two columns have an error each. Next we check whether the syndrome is a linear combination of the fifth and sixth columns of $H$ or not. A simple calculation tells us that the syndrome is indeed a linear combination of these two columns of $H$. Thus we get $\overline{e}_5=\ob$ and $\overline{e}_6=\ob$.


\[
{\setlength{\arraycolsep}{1.785pt}
\renewcommand{\arraystretch}{.5}
 H {\bf \overline{y}}^T = \left[\begin{array}{c}
                          0\\
                          0 \\
                          0 \\
                          0 \\
                         \ob
                   \end{array}\right]
                 = \overline{e}_5 \left[\begin{array}{c}
                          0 \\
                          1 \\
                          1 \\
                          0 \\
                          1
                   \end{array}\right]
                 +\overline{e}_6 \left[\begin{array}{c}
                          0 \\
                          1 \\
                          1 \\
                          0 \\
                          0
                   \end{array}\right]=H {\bf \overline{e}}^T.
}\]


So we can get a correct projection
\ben
\begin{split}
{\bf y'}&={\bf y}+{\bf e}\\
&=(\om,\om,\ob,\ob,\om,\om,1,1,0,0)+(0,0,0,0,\om,\om,0,0,0,0)\\
&=(\om,\om,\ob,\ob,{\bf 0},{\bf 0},1,1,0,0).
\end{split}
\een

We use ${\bf y'}$ to correct ${\bf v}$ by changing  two elements in two  columns where the errors
occurred  according to the condition that all the columns and top row have the same parity.  Then we say   ${\bf v}$ is uniquely decoded as a ``codeword of $C_{40,1}^{DE}$'' and  exit.
\end{itemize}

}

\end{example}

\begin{example}[Case IV]
\label{CaseIV} Seven columns have one parity and remaining three
columns have the other parity.

\[
{\setlength{\arraycolsep}{1.785pt}
\renewcommand{\arraystretch}{.5}
{\bf v} = \begin{array}{cccccccccccc}
        &  &  1&  2&  3&  4&  5&  6&  7&  8&  9& 10\\ \hline
       0&  &  1&  1&  1&  1&  1&  1&  1&  1&  1&  1\\
       1&  &  1&  0&  1&  1&  1&  1&  1&  1&  1&  1 \\
     \om&  &  1&  1&  1&  0&  0&  1&  0&  1&  0&  1 \\
     \ob&  &  0&  0&  0&  1&  0&  1&  1&  0&  1&  0 \\
      \hline
 {\bf y}&  &\ob&\om&  \ob&\om&  1&  0& \om&\ob&\om& \ob \\
      \end{array}
}\]

{\em
{\bf{(a) Representation Decoding }}
\begin{itemize}
\item[Step 1] Compute the parities of the columns of ${\bf v}$ and determine which
case of Table \ref{tab:t2} we are in. Since seven columns have one parity and remaining three
columns have the other parity,  this is the case IV.  Thus  we move to the Step 2-IV-(i).

\item[Step 2-IV-(i)] If the parity of ${\bf v}$ is the case IV, then consider it as a case IV-(i).  Compute the projection of ${\bf v}$ onto $GF(4)^{10}$, call it ${\bf y}$. Since we know the error column positions, it is easy to find   the closest vector ${\bf y'}$ to ${\bf y}$  in Table 1  by changing three elements of ${\bf y}=(\ob,-,\ob,\om,-,-,\om,\ob,\om,\ob)$ after the action of the automorphism groups and scalar multiplication.

We see from Table \ref{tab:t1}
that ${\bf y'}=(\ob,\om,\ob,\om,\om,\ob,\om,\ob,\om,\ob)$ (Type 7 of Table \ref{tab:t1}). Hence  we use ${\bf y'}$ to correct ${\bf v}$ by changing  three elements in three  columns where the errors
occurred  according to the condition that all the columns and top row have the same parity as follows:
\[
{\setlength{\arraycolsep}{1.785pt}
\renewcommand{\arraystretch}{.5}
{\bf v} = \begin{array}{cccccccccccc}
        &  &  1&  2&  3&  4&  5&  6&  7&  8&  9& 10\\ \hline
       0&  &  1&  {\bf 0}&  1&  1&  1&  1&  1&  1&  1&  1\\
       1&  &  1&  0&  1&  1&  1&  1&  1&  1&  1&  1 \\
     \om&  &  1&  1&  1&  0&  0&  1&  0&  1&  0&  1 \\
     \ob&  &  0&  0&  0&  1&  {\bf 1}& {\bf 0}&  1&  0&  1&  0 \\
      \hline
 {\bf y'}&  &\ob&\mbox{\boldmath \om}&  \ob&\om& \mbox{\boldmath \om}& \mbox{\boldmath \ob}& \om&\ob&\om& \ob \\
      \end{array}.
}\]

Now ${\bf v}$ is decoded as a ``codeword of $C_{40,1}^{DE}$'' and we exit.

\end{itemize}

{\bf{(b) Syndrome Decoding }}
\begin{itemize}

\item[Step 1] Same as the Step 1 in representation decoding.

\item[Step 2-IV-(i)] If the parity of ${\bf v}$ is the case IV $(7,\bar{3})$, then consider it as a case IV-(i).  Compute the projection of ${\bf v}$ onto $GF(4)^{10}$, call it ${\bf y}$.
Since we know the second, fifth, and sixth columns of ${\bf v}$ are the only ones with even parities, we know these three columns have an error each. Next we check whether the syndrome is a linear combination of the second, fifth, and sixth columns of $H$ or not. A simple calculation tells us that the syndrome is indeed a linear combination of these three columns of $H$.


\[
{\setlength{\arraycolsep}{1.785pt}
\renewcommand{\arraystretch}{.5}
 H {\bf \overline{y}}^T = \left[\begin{array}{c}
                          0\\
                          0\\
                          0\\
                          0\\
                        \om
                   \end{array}\right]
                 = \overline{e}_2 \left[\begin{array}{c}
                          1 \\
                          0 \\
                          0 \\
                          0\\
                         0
                   \end{array}\right]
                 +\overline{e}_5 \left[\begin{array}{c}
                          0 \\
                          1 \\
                          1 \\
                          0\\
                          1
                   \end{array}\right]
                 +\overline{e}_6 \left[\begin{array}{c}
                          0 \\
                          1 \\
                          1 \\
                          0\\
                          0
                   \end{array}\right]=H {\bf \overline{e}}^T.
}\]

We do some calculations from the first entry of the syndrome as follows:

\begin{equation*}
\begin{split}
0 &= \overline{e}_2 \cdot 1+ \overline{e}_5 \cdot 0 +\overline{e}_6\cdot 0,\\
\Rightarrow \overline{e}_2 &=0.\\
\end{split}
\end{equation*}
Now we know $\overline{e}_2 =0$.
Calculating on the fifth entry of the syndrome, we get

\begin{equation*}
\begin{split}
\om &= \overline{e}_5 \cdot 1+ \overline{e}_6\cdot 0, \\
\Rightarrow \overline{e}_5& = \om. \\
\end{split}
\end{equation*}
Now we know $\overline{e}_5 =\om $.
Calculating on the third entry of the syndrome, we get
\begin{equation*}
\begin{split}
0& =  \om  \cdot 1 +\overline{e}_6\cdot 1,\\
\Rightarrow \overline{e}_6 &=\om. \\
\end{split}
\end{equation*}
Now we know $\overline{e}_6 =\om$ as well.

Thus we know the error vector from $\overline{e}_2=0,\overline{e}_5 =\om$, and $\overline{e}_6=\om$  and get a correct projection ${\bf y'}$ as follows:
\ben
\begin{split}
{\bf y'}&={\bf y}+{\bf e}\\
&=(\ob,\om,\ob,\om,1,0,\om,\ob,\om,\ob)+(0,0,0,0,\ob,\ob,0,0,0,0)\\
&=(\ob,\mbox{\boldmath \om},\ob,\om,\mbox{\boldmath \om},\mbox{\boldmath \ob},\om,\ob,\om,\ob).
\end{split}
\een

Now we use ${\bf y'}$ to correct ${\bf v}$ by changing  three elements in three  columns where the errors
occurred  according to the condition that all the columns and top row have the same parity.
Finally  ${\bf v}$ is decoded as a ``codeword of $C_{40,1}^{DE}$'' and we exit.

\end{itemize}

 }
\end{example}


\end{document}